\documentclass[11pt]{article}
\usepackage{amsfonts,amssymb,amsmath,latexsym,ae,aecompl}
\usepackage{graphics}
\usepackage{latexsym}
\usepackage{algorithm}
\usepackage{algorithmic}
\usepackage{tikz}
\usepackage{caption}
\usepackage{subfig}
\newcommand{\qed}{\hfill$\Box$}
\newcommand{\cU}{{\cal U}}
\newcommand{\cF}{{\cal F}}
\newcommand{\cG}{{\cal G}}
\newcommand{\cA}{{\cal A}}
\newcommand{\cV}{{\cal V}}
\newcommand{\cE}{{\cal E}}

\newenvironment{proof}{\noindent {\bf Proof.}}{\qed}

\newtheorem{theorem}{Theorem}[section]
\newtheorem{lemma}{Lemma}[section]

\begin{document}

\baselineskip 0.2in
\parskip      0.1in
\parindent    0em

\bibliographystyle{plain}

\title{{\bf Explorable Families of Graphs} }

\author{
Andrzej Pelc  \footnotemark[1]
}

\date{ }
\maketitle
\def\thefootnote{\fnsymbol{footnote}}

\footnotetext[1]{
\noindent
D\'{e}partement d'informatique, Universit\'{e} du Qu\'{e}bec en Outaouais,
Gatineau, Qu\'{e}bec J8X 3X7,
 Canada. {\tt pelc@uqo.ca}.
Research supported in part by NSERC  Discovery Grant 8136 -- 2013  and by the
Research Chair in Distributed Computing of the
Universit\'{e} du Qu\'{e}bec en Outaouais.
}

\begin{abstract}
Graph exploration is one of the fundamental tasks performed by a mobile agent in a graph. An $n$-node graph has unlabeled nodes,
and all ports at any node of degree $d$ are arbitrarily numbered $0,\dots, d-1$. A mobile agent, initially
situated at some starting node $v$, has to visit all nodes of the graph and stop. In the absence of any initial knowledge of the graph
the task of deterministic exploration is often impossible. On the other hand, for some families of graphs it is possible to design deterministic exploration algorithms
working for any graph of the family. We call such families of graphs {\em explorable}. Examples of explorable families are all finite families of graphs, as well as 
the family of all trees. 

In this paper we study the problem of which families of graphs are explorable. We characterize all such families, and then ask the question whether there exists
a universal deterministic algorithm that, given an explorable family of graphs, explores any graph of this family, without knowing which graph of the family is being explored.
The answer to this question turns out to depend on how the explorable family is given to the hypothetical universal algorithm. If the algorithm can get the answer to any yes/no
question about the family, then such a universal algorithm can be constructed. If, on the other hand, the algorithm can be only given an algorithmic description of the
input explorable family, then such a universal deterministic algorithm does not exist. 

\vspace*{0.5cm}


\vspace*{5cm}

\end{abstract}

\pagebreak

\section{Introduction}

Network exploration is one of the fundamental tasks performed by mobile agents in networks.
Depending on the application, the mobile agent may be 
a software agent that has to collect data placed at nodes of a communication network, or it may be a mobile robot
collecting samples of ground in a contaminated building or mine whose corridors form links of a network, with
corridor crossings represented by  nodes.

The network is modeled as a finite simple connected
undirected graph $G=(V,E)$ with $n$ nodes, called {\em graph} in the sequel. The number $n$ of nodes is called the {\em size} of the graph.
Nodes are unlabeled, and all ports at any node of degree $d$ are arbitrarily numbered $0,\dots, d-1$.
The agent is initially situated at a starting node $v$ of the graph. When the agent located at a current node $u$ gets to a neighbor $w$ of $u$
by taking a port $i$, it learns the port $j$ by which it enters node $w$ and it learns the degree of $w$.
The agent has to visit all nodes of the graph and stop. We assume that the agent is computationally unbounded (it is modeled as a Turing machine) and cannot mark the visited nodes.

It is well-known that,  without any a priori information, the task of deterministic exploration is impossible to perform in arbitrary graphs. In fact, it is impossible even in quite simple and restricted families of graphs, such as the 
class of rings in which ports at all nodes are numbered 0,1 in clockwise order. Even if the agent knows that it is in some such ring (but does not know which), it cannot learn its size. 
If there existed a deterministic exploration algorithm for the class of such rings,  not using any a priori knowledge,
then the agent would have to stop after some $t$ steps in every ring, and hence it would fail to explore a $(t+2)$-node ring.

On the other hand, there exist classes of graphs for which deterministic exploration of any graph in the class is possible without knowing in which graph of the family the agent operates. Such are, for example, all finite classes of graphs.
Knowing such a class, the algorithm can find an upper bound $N$ on the size of all graphs in the family, and then apply , e.g., the algorithm from \cite{Re}, based on universal exploration sequences, that
visits all nodes of any graph of size at most $N$, regardless of the starting node.  On the other hand, there are also infinite families of graphs that are possible to explore without any initial information.
Such is, for example, the family of all trees. Indeed, any tree can be visited using the {\em basic walk} that consists in leaving the starting node by port 0 and then leaving every node $w$ by port $(i+1) \mod d$,
where $i$ is the port by which the agent entered node $w$ and $d$ is the degree of $w$. Performing such a walk, the agent realizes when it made the full tour of the tree and got back to the starting node, where it stops.

The aim of this paper is to study families of graphs that have the property that any graph in the family can be deterministically explored without knowing in which graph the agent operates. We adopt the following definition.

\begin{quotation}
A family $\cal F$ of graphs is {\em explorable}, if there exists a deterministic algorithm $\cal A(\cF)$
 dedicated to this family, such that a mobile agent that executes algorithm $\cA (\cF)$ starting from any node $v$ of any graph $G\in \cal F$, visits all nodes of $G$ and stops.
\end{quotation}

{\bf Our results.}
We give an exact characterization of explorable families of graphs by formulating a condition $C$ with the following properties. Given a family $\cal F$ of graphs that does not satisfy condition $C$,
 no deterministic algorithm can explore all graphs of $\cal F$. On the other hand, given any family $\cal F$ of graphs satisfying condition $C$ we construct a deterministic algorithm $\cal A(\cal F)$
that explores all graphs of $\cal F$. 

The above algorithm $\cal A(\cal F)$ used to explore graphs of a family $\cal F$ that has property $C$ is {\em dedicated} to the family $\cal F$, i.e., it works only for graphs from $\cal F$, and for different
families $\cal F$ different  algorithms $\cal A(\cal F)$ are used. Hence it is natural to ask if there exists
a {\em universal} deterministic algorithm $\cal U$ that, given an explorable family $\cal F$ of graphs (i.e.. any family satisfying condition $C$), explores any graph of this family, without knowing which graph of the family it is exploring.
The answer to this question turns out to depend on how the explorable family $\cal F$ is given to the hypothetical universal algorithm. (Since interesting explorable families are infinite, the input of $\cal U$ cannot be given as a finite object all at once, e.g., coded as a finite binary string). If the universal algorithm can get the answer to any yes/no
question about the family, then such a universal deterministic algorithm $\cal U$ can be constructed. If, on the other hand, the universal algorithm can be only given an algorithmic description of the
input explorable family, then such a universal deterministic algorithm does not exist.

{\bf Related work.}
Exploration of unknown environments by mobile agents
has been studied for many decades (cf. the survey \cite{RKSI}).
The explored environment can be modeled in two distinct ways:
either as a subset of the plane, e.g., an
unknown terrain with convex obstacles \cite{BRS},
or a room with polygonal \cite{DKP} or rectangular \cite{BBFY} obstacles, or
as a graph, assuming that the agent may only move along its edges. The graph
model can be further split into two different scenarios.
One of them assumes that the graph is directed, in which case the agent can move only from
tail to head of a directed edge  \cite{AH,BFRSV,BS}. The other scenario assumes that the graph is undirected 
and the agent can traverse edges in both directions  \cite{ABRS,BRS2,DKK,PaPe}.
Some authors impose further restrictions on the moves of the agent.
In  \cite{ABRS,BRS2} it is assumed that the agent has a restricted tank,
and thus has to periodically return to the base for refueling, while the authors of \cite{DKK}  assume that the agent is attached to the
base by a cable of restricted length.

An important direction of research concerns
exploration of anonymous graphs.
In this case it is impossible to perform exploration with termination of
arbitrary graphs in the absence of any a priori knowledge of the graph, if no marking of nodes is allowed.
Hence some authors \cite{BFRSV,BS}
allow {\em pebbles} which the agent can drop on nodes to recognize already visited ones, and
then remove them and drop them in other nodes. A more restrictive scenario assumes that 
a stationary token is placed at the starting node of the agent \cite{CDK,PeTi}.
Exploring
anonymous graphs without the possibility of marking nodes (and thus possibly without stopping)
is investigated, e.g., in \cite{DFKP,FI}.
In these papers the authors
concentrate attention on the minimum amount of memory sufficient
to carry out exploration. If marking of nodes is precluded, some knowledge about the graph is required in order to guarantee stopping after exploration,
e.g., an upper bound on its size \cite{AKLLR,CDK,Re}.

In \cite{FIP2,GP}, the authors study the problem of the minimum size of information that has to be given to the mobile agent in order to perform fast exploration.
In  \cite{FIP2}, only exploration of trees is considered, and the algorithm performance is measured using the competitive approach.
In \cite{GP}, exploration of arbitrary graphs is studied, and the performance measure is the order of magnitude of exploration time.

\section{Characterization of explorable families}

In this section we provide a necessary and sufficient condition on the explorability of a family of graphs. From now on, we restrict attention to infinite families of graphs, since, as mentioned in the introduction, all finite families are trivially explorable. In order to formulate the condition we need the notion of a {\em truncated view} from a node $v$ in a graph $G$. Let $G$ be any graph, $v$ a node in this graph and $k$ a natural number. The truncated view from $v$ in $G$ of depth $k$, denoted $\cV(v,G,k)$, is the tree of all simple paths of length at most $k$, starting from node $v$ and coded as sequences of port numbers, where the rooted tree structure is defined by the prefix relation of sequences. 
This definition is equivalent  to that from \cite{YK}. It follows from the above paper that the information that a mobile agent starting at node $v$ in graph $G$ can obtain after $k$ steps is ``included'' in
$\cV(v,G,k+1)$  in the following sense. For any deterministic algorithm $\cA$ that works in graphs $G$ and $G'$ and any mobile agents $A_1$ and $A_2$ executing this algorithm, starting, respectively, from
node $v$ in $G$ and from node $v'$ in $G'$, such that $\cV(v,G,k+1)=\cV(v',G',k+1)$, the behaviors of agents $A_1$ and $A_2$ during the first $k$ steps of these executions, i.e., their trajectories and possible decisions to stop, are identical.

Since the set of (finite) graphs is countable, every (infinite) family $\cF$ of graphs can be represented as a sequence $\{G_i : i\geq 1\}$, ordered so that no two graphs $G_i$ and $G_j$ are port-preserving isomorphic, and the sizes of the graphs are non-decreasing. In order to make the enumeration unambiguous, graphs of the same size are ordered lexicographically, using some fixed graph representation. The resulting ordering will be called {\em canonical} and used from now on. 

The following condition $C$ concerns a family $\cF=\{G_i : i\geq 1\}$ of graphs:

\begin{quotation}
For every $i\geq 1$ and every node $v$ in $G_i$, there exist positive integers $k,m$, such that $\cV(v,G_i,k)$ is different from truncated views
of depth $k$ from all nodes in all graphs $G_j$, for $j>m$.
\end{quotation}

We now proceed to the proof that $C$ is a necessary and sufficient condition on the explorabilty of $\cF$. We first prove the necessity.

\begin{lemma}\label{nec}
A family $\cF=\{G_i : i\geq 1\}$ that does not satisfy the condition $C$ is not explorable.
\end{lemma}

\begin{proof}
Suppose that the condition $C$ is not satisfied for a family $\cF=\{G_i : i\geq 1\}$ of graphs. This implies that there exists a positive integer $i$
and a node $v$ in  graph $G_i$, such that for all positive integers $k, m$, there exists an index $j(k,m)>m$ and a node $v(k,m)$ in the graph $G_{j(k,m)}$, satisfying the equality $\cV(v(k,m), G_{j(k,m)},k)=\cV(v, G_i,k)$.

Suppose, for a contradiction, that the family $\cF$ is explorable. Hence there exists an algorithm $\cA(\cF)$ that explores any graph of the family, starting from any node. Consider the execution of this algorithm by an agent $A_1$ on graph $G_i$, starting from node $v$. For some integer $x$, after $x$ steps, the agent explores the graph $G_i$ and stops. Take $k=x+1$ and any $m$ such that all graphs $G_j$, for $ j>m$, have sizes larger than $x+1$. Consider an index $j(k,m)>m$ and a node $v(k,m)$ in $G_{j(k,m)}$  for which $\cV(v(k,m), G_{j(k,m)},k)=\cV(v, G_i,k)$.
Now consider the execution of algorithm  $\cA(\cF)$ by an agent $A_2$ starting at node $v(k,m)$ in $G_{j(k,m)}$. Since $\cV(v(k,m), G_{j(k,m)},k)=\cV(v, G_i,k)$, the behavior of this agent during the first $x$ steps must be the same as the behavior of agent $A_1$. It follows that agent $A_2$ stops after $x$ steps in the graph $G_{j(k,m)}$ as well. However, this graph has size larger than $x+1$ and hence it cannot be explored in $x$ steps. This contradiction proves the lemma.
\end{proof}

We next proceed to the proof of the sufficiency of condition $C$. In order to do this we will construct an algorithm $Explo(\cF)$, dedicated to a family of graphs $\cF=\{G_i : i\geq 1\}$ satisfying condition $C$, such that $Explo(\cF)$ explores all graphs in this family. For any 
$i\geq 1$ and every node $v$ in $G_i$, let $(k(v,i), m(v,i))$ be the lexicographically first couple of integers $(k,m)$, such that $\cV(v,G_i,k)$ is different from truncated views
of depth $k$ from all nodes in all graphs $G_j$, for $j>m$. Call the integer $k(v,i)$ the {\em depth witness} of $(v,i)$ and call the integer $m(v,i)$ the {\em range witness} of $(v,i)$.  
For every graph, we define a {\em non-backtracking} path as a path in which the agent
never exits a node by a port $p$ immediately after entering it by the port $p$. The {\em reverse} of a path $P=(w_1,w_2,\dots , w_s)$ is the path $\overline{P}=(w_s,w_{s-1}, \dots, w_1)$.

We first describe Procedure {\tt Check} $(v,i)$, for a positive integer $i$ and for a node $v$ in $G_i$. The aim of this procedure is to check
whether the truncated view of depth $k(i,v)$ from the unknown initial position of the agent in an unknown graph from $\cF$ is equal to the truncated view $\cV(v,G_i, (k(i,v))$. To this end, the agent traverses (in lexicographic order) all non-backtracking maximal paths from $\cV(v,G_i, (k(i,v))$   starting at its initial node and returning to it after traversing each path, using the reverse of this path. In the case when a given non-backtracking path 
is impossible to traverse, either because the agent enters earlier a node of degree 1, or because the entry port at a given edge is different from that in $\cV(v,G_i, (k(i,v))$, then the agent interrupts the traversal of this path and returns using the reverse of the path traversed till this point.

There are two possible terminations of 
Procedure {\tt Check} $(v,i)$. The first possibility is that all non-backtracking paths of length $k(v,i)$ are identical as in $\cV(v,G_i, (k(v,i))$. This is called the {\em success} of {\tt Check} $(v,i)$. The second possibility, called the {\em failure} of {\tt Check} $(v,i)$, is that the above condition is not satisfied.

The next procedure is used to find a couple $(v,i)$ for which {\tt Check} $(v,i)$ is terminated with success. Such a couple must exist: indeed,
if the agent starts from node $v$ in the graph $G_i$, then  {\tt Check} $(v,i)$ is terminated with success.

\begin{center}
\fbox{
\begin{minipage}{12cm}

{\bf Procedure} {\tt Find Success}\\

\vspace*{0.5cm}

$i:=1$\\
$result := failure$\\
{\bf while} $result = failure$ {\bf do}\\
\hspace*{1cm}{\bf for all} nodes $v$ in $G_i$ {\bf do}\\
\hspace*{2cm}{\tt Check}  $(v,i)$\\
\hspace*{2cm} {\bf if} {\tt Check}  $(v,i)$ terminated with success {\bf then}\\
\hspace*{3cm}$result := success$\\
\hspace*{1cm}$i:=i+1$\\
{\bf return} parameters $v$ and $i$ for which the variable $result$ changed from $failure$ to $success$.

\end{minipage}
}
\end{center}

Let $R(N)$ be the procedure from \cite{Re} based on universal exploration sequences, that
visits all nodes of any graph of size at most $N$, and stops, regardless of the starting node. Now the algorithm $Explo(\cF)$ can be succinctly formulated as follows.

\begin{center}
\fbox{
\begin{minipage}{12cm}

{\bf Algorithm} $Explo(\cF)$\\

\vspace*{0.5cm} 

{\tt Find Success}\\
$(v,i):=$ parameters returned by procedure {\tt Find Success}\\
$M:=$ the maximum size of all graphs in the family $\{G_t : t\leq m(v,i)\}$\\
$R(M)$

\end{minipage}
}
\end{center}

\begin{lemma}\label{suf}
For every family $\cF$ of graphs, satisfying condition $C$, Algorithm $Explo(\cF)$ correctly explores any graph of the family $\cF$, starting at any node of this graph.
\end{lemma}

\begin{proof}
The {\bf while} loop in Procedure {\tt Find Success} must be exited at some point, i.e., the variable $result$ must be set to $success$. This happens at the latest in the execution of Procedure {\tt Check}  $(v,i)$, where the graph from the family $\cF$ in which the agent is operating is $G_i$ and the starting node of the agent is $v$. Indeed, in this case, the truncated view from the initial position of the agent at any depth $h$ is identical with the truncated view $\cV(v,G_i, h)$. Hence the Procedure {\tt Find Success} terminates and the agent learns  parameters $v$ and $i$ for which the variable $result$ changed from $failure$ to $success$. By the definition of $m(v,i)$, the agent learns that it is  in one of the graphs from the (finite) family $\{G_t : t\leq m(v,i)\}$. It follows that the execution of procedure $R(M)$, where $M$ is the maximum size of all graphs in this family, must result in the exploration of the graph in which the agent operates, regardless of which graph of the family  $\{G_t : t\leq m(v,i)\}$ it is and regardless of which node  of this graph is the starting node of the agent. This proves the lemma.
\end{proof}

Lemmas \ref{nec} and \ref{suf} imply the following characterization result.

\begin{theorem}
A family of graphs is explorable if and only if it satisfies condition $C$.
\end{theorem}

\section{Universal exploration algorithm}

In this section we use the characterization from Section 2 to investigate the following problem. Does there exist a {\em universal} algorithm
$\cU$, which when given as input an explorable family $\cF$ of graphs, explores any graph of this family, starting from any initial node in it?
Note that Algorithm $Explo(\cF)$ from the preceding section was {\em dedicated} to the exploration of graphs of one particular explorable family $\cF$, i.e., it was supposed to work only for graphs of this family. Consequently, important information about the family, such as functions
$k(v,i)$ and $m(v,i)$, could be included in the text of the algorithm. For a universal algorithm this is not the case: in contrast to dedicated algorithms, a universal algorithm must gain sufficient knowledge about the input explorable family $\cF$, in order to be able to successfully explore any graph of this family.

Here comes the subtle issue of how the input explorable family is given. First recall that we may restrict attention to infinite families, as finite 
families are trivial to explore by applying the procedure $R(N)$ from \cite{Re}, where $N$ is an upper bound on the sizes of all graphs in the family. For a finite family, it can be given to the universal algorithm as a single finite input object, the algorithm can find $N$, apply $R(N)$ and we are done. By contrast, in the case of infinite explorable families $\cF$, the family cannot be given to the hypothetical universal algorithm $\cU$ as a single input object all at once. How then could it be given?

It seems reasonable to assume that the universal algorithm should be able to get knowledge about the input family $\cF$ ``piece by piece'',
i.e., it should be able to get items of information about this family as responses to queries. This idea can be implemented in at least two ways. We start with the more liberal way that intuitively allows the algorithm to get an answer to any yes/no query about the input family $\cF$.
This can be formalized as follows. Consider the set $\cal X$ of all infinite families of finite graphs. Observe that, while each family in $\cal X$ is countable, the set $\cal X$ which is the set of all these families is uncountable, but this fact has no impact on our formalization.
Consider any definable subfamily $\Xi$ of $\cal X$, i.e., a family $\Xi=\{\cG \in {\cal X} : \cG$ satisfies $\Phi\}$, where $\Phi$ is some set-theoretic predicate. The questions that the universal algorithm is allowed to ask are of the form: ``Is the input family $\cF$ an element of $\Xi$?''
This is of course equivalent to asking ``Does the input family $\cF$ satisfy the predicate $\Phi$?'' Such questions can be asked by the universal algorithm, using  all possible predicates $\Phi$, one at a time, and it is assumed that the algorithm will obtain a truthful answer to any such question. This formalization can be thought of as using an oracle that knows everything about the input family but answers only yes or no.
Examples of questions that the algorithm can ask are: ``Are there infinitely many planar graphs in $\cF$?'', ``Is the 17-th graph in the canonical order of $\cF$  a tree? or ``Does there exist a tree in $\cF$?''. 

It should be noted that although the allowed queries are only of yes/no type, they are very powerful, as the universal algorithm may get some information about the entire infinite input family all at once, as in the query ``Does there exist a tree in $\cF$?''.  A negative answer to such a query could not be obtained
by looking at any finite part of $\cF$. We will show that this powerful feature allows us to construct a correct universal exploration algorithm.
Before doing it, we need some preparation.

Let $\{H_i: i \geq 1\}$ be the canonical enumeration of the family $\cal X$. Consider an explorable family $\cF$ that is the input to the universal algorithm that we are going to describe. We first describe the procedure
{\tt Find i-th graph}  that returns the graph that is the $i$-th element in the canonical enumeration of $\cF$. The procedure asks the questions
``Is the $i$-th element in the canonical enumeration of $\cF$ equal to $H_j$?'', for $j=1,2,...$, until the answer yes is obtained, and returns
$H_j$ for which the positive answer is obtained.

The aim of the next two procedures is finding, respectively, the depth witness and the range witness of $(v,i)$, where $v$ is a node in the
$i$-th element in the canonical enumeration of $\cF$. Procedure {\tt Find the depth witness of} $(v,i)$ asks questions ``Is the depth witness of 
$(v,i)$ equal to $j$, for $j=1,2,...$, until the answer yes is obtained, and returns the integer $j$ for which the positive answer is obtained.
Similarly, procedure {\tt Find the range witness of} $(v,i)$ asks questions ``Is the range witness of 
$(v,i)$ equal to $j$, for $j=1,2,...$, until the answer yes is obtained, and returns the integer $j$ for which the positive answer is obtained.

We will now modify the procedure {\tt Check} $(v,i)$ from Section 2 to make it work in the context of the universal algorithm.
The modification, for any $(v,i)$ consists in first applying procedure {\tt Find i-th graph} and then applying procedure {\tt Find the depth witness of} $(v,i)$. Suppose that the first procedure returns the graph $G_i$ and the second procedure returns the integer $k(v,i)$. The rest of procedure  {\tt Check} $(v,i)$ is as in Section 2. Again, the procedure may terminate with success or failure, defined as previously.
Procedure {\tt Find Success} is as before, using the modified version of  {\tt Check} $(v,i)$. We will call it {\tt Universal Find Success}. Now our universal algorithm can be formulated as follows, assuming that the input explorable family given to the oracle that answers queries is $\cF$.

\begin{center}
\fbox{
\begin{minipage}{12cm}

{\bf Algorithm} {\tt Universal Exploration}\\

\vspace*{0.5cm}
{\tt Universal Find Success}\\
$(v,i):=$ parameters returned by procedure {\tt Find Success}\\
{\tt Find range witness of} $(i,v)$\\
Let $m(i,v)$ be the range witness of $(i,v)$\\
{\bf for} $t=1$ {\bf to} $m(i,v)$ {\bf do}\\
\hspace*{1cm}{\tt Find the t-th graph}\\
\hspace*{1cm}Let $G_t$ be the graph returned by procedure {\tt Find the t-th graph}\\
$M:=$ the maximum size of all graphs in the family $\{G_t : t\leq m(v,i)\}$\\
$R(M)$

\end{minipage}
}
\end{center}

The correctness of Algorithm {\tt Universal Exploration} follows from the fact that it correctly finds the depth and range witnesses, as well
as the graphs $\{G_t : t\leq m(v,i)\}$. Other than that, the algorithm works like the dedicated algorithm $\cA(\cF)$, and thus achieves exploration of any graph in the family. Hence we have the following theorem.

\begin{theorem}
Algorithm {\tt Universal Exploration} correctly explores any graph of an explorable family $\cF$, starting at any node of this graph, if this family is given as input to an oracle that can answer all yes/no queries about it.
\end{theorem}

The capability of getting an answer to any yes/no query about the input explorable family is very strong. It assumes the existence of an oracle that has a ``magical'' complete insight in this family. It can be argued that such an oracle could not exist in practice, and thus it would be desirable to design a universal exploration algorithm to which input explorable families would be given in a way possible to implement realistically. Here comes the second natural way in which a potential universal algorithm could get information about the input family. Suppose that the input explorable family $\cF$ is recursively enumerable and let $\cE(\cF)$ be the enumeration algorithm. More precisely, the algorithm
$\cE(\cF)$, given a positive integer $i$ as input, returns the $i$-th graph of $\cF$ in the canonical enumeration. The second natural way of providing information about the family $\cF$ to the hypothetical universal exploration algorithm $\cU$ would be to give to $\cU$ the text of algorithm
$\cE(\cF)$ as input. Then the algorithm $\cU$ would be able, for any positive integer $i$, to run $\cE(\cF)$ on $i$ and get the $i$-th graph of $\cF$ in the canonical enumeration, returned by $\cE(\cF)$. For simplicity, we may assume that finding this $i$-th graph is done in one step.
This is reminiscent of the definition of Turing reducibility in which an algorithm $A_1$ reducible to $A_2$ may run $A_2$ on some input and receive the output in one step. In any case, in this paper we are not concerned with efficiency of exploration, only with the feasibility of this task.

We will say that a universal exploration algorithm {\em processes algorithmic input},  if it works as described above. Such an algorithm would be able to explore any graph of any (recursively enumerable) explorable family, given to it in this algorithmic way, without the help of any oracle.
Unfortunately, we have the following negative result.

\begin{theorem}\label{no-univ}
There does not exist a universal exploration algorithm that processes algorithmic input, which correctly explores any graph of any recursively enumerable explorable family.
\end{theorem}

\begin{proof}
Suppose that there exists a universal exploration algorithm $\cU$ that processes algorithmic input, which correctly explores any graph of any recursively enumerable explorable family. Denote by $R_k$, for $k\geq 3$, the ring of size $k$ with ports at all nodes numbered 0,1 in the clockwise order. Let $C_k$ denote the graph resulting from the ring $R_k$ by attaching one node of degree 1 to one of the nodes of $R_k$. By definition, the edge joining the single node $v$ of degree 3 with the single node of degree 1 corresponds to port number 2 at $v$.
 In $C_k$ we will say that an agent goes clockwise if it leaves a node by port 1. For $i \geq 1$, let $G_i$ be the graph $C_{i+2}$, and consider the family $\cF=\{G_i : i\geq 1\}$ of graphs. Since the sizes of graphs $G_i$ are strictly increasing, this is the canonical enumeration of $\cF$.

We first observe that the family $\cF$ is explorable. The (dedicated) exploration algorithm $\cA(\cF)$ can be simply formulated as follows.

If the starting node is of degree 1 then take port 0, go clockwise until getting to a node of degree 3 and stop.\\
If the starting node is of degree 3 then take port 2, take port 0, go clockwise until getting to a node of degree 3 and stop.\\
If the starting node is of degree 2 then go clockwise until the second visit at a node of degree 3, take port 2 and stop.

Suppose that $\cE(\cF)$ is an enumeration algorithm corresponding to $\cF$, given to $\cU$ as input. Consider the execution $E_1$ of
$\cU$ with this input, where the agent is initially placed at the single node $v$ of degree 3 in graph $G_1$. Suppose that in execution $E_1$ the agent stops after $k$ steps. Let $r$ be the largest integer for which algorithm $\cE(\cF)$ was called in the execution $E_1$ of $\cU$.

For any positive integer $j$, define the following graph $D_j$ (cf. Fig. \ref{dessin} at the end of the paper). Consider the ring $R_{6j}$. Attach to every third node of $R_{6j}$ a distinct node of degree 1. Finally, to one of the $2j$ resulting nodes of degree 3 attach another node of degree 1. Thus graph $D_j$ has 1 node of degree 4 and $2j-1$ nodes of degree 3, partitioning the ring $R_{6j}$ from which the construction started into $2j$ segments of length 3. Moreover, $D_j$ has $2j+1$ nodes of degree 1. The port at each node of degree 3 corresponding to the edge joining it to a node of degree 1 has number 2.
The ports at the single node of degree 4 corresponding to the edges joining it to nodes of degree 1 have numbers 2 and 3.

We now define the graphs $H_i$, for $i \geq 1$, as follows. For $1\leq i\leq r$, let $H_i=G_i$; for $i>r$, let $H_i=D_i$. Finally, we consider the family $\cF^*=\{H_i : i\geq 1\}$ of graphs. Since the sizes of graphs $H_i$ are strictly increasing, this is the canonical enumeration of $\cF^*$.
We show that the family  $\cF^*$ is explorable. In order to formulate the (dedicated) algorithm $\cA(\cF^*)$, we first describe the following procedure that explores any graph $D_j$ starting from a node of degree 3. (For this purpose, 
repeating until the second visit at a node of degree 4 would be enough, but we need the third visit to make it work within algorithm $\cA(\cF^*)$).

\begin{center}
\fbox{
\begin{minipage}{9.5cm}

{\bf Procedure} {\tt Go around}

\vspace*{0.5cm}

{\bf repeat}\\
\hspace*{1cm}take port 1, take port 1, take port 2, take port 0\\
{\bf until} the third visit at a node of degree 4\\
take port 3 and stop.

\end{minipage}
}
\end{center}

Now algorithm $\cA(\cF^*)$ can be described as follows. Its high level idea is to go clockwise around the ring in any graph $H_j$, sufficiently long to see if $j$ is at  most $r$ or larger than $r$. In the first case, algorithm $\cA(\cF)$ is applied because the graph is $G_j$, and in the second case procedure {\tt Go around} is used to terminate exploration because the graph is $D_j$.

\begin{center}
\fbox{
\begin{minipage}{11cm}

{\bf Algorithm} $\cA(\cF^*)$

\vspace*{0.5cm}

{\bf if} the starting node is of degree 1 {\bf then} take port 0\\ 
Go clockwise for $r+1$ steps.\\ 
{\bf if} no node of degree 4 is visited {\bf then} apply the algorithm $\cA(\cF)$\\
{\bf else}\\ 
\hspace*{1cm}go clockwise to the closest node of degree 3\\ 
\hspace*{1cm}{\tt Go around}

\end{minipage}
}
\end{center}

Let $m=\max(k,r)+1$ and suppose that $\cE(\cF^*)$ is an enumeration algorithm corresponding to $\cF^*$, given to $\cU$ as input.
Consider the execution $E_2$ of algorithm $\cU$ with this input, where the agent is initially placed at the node $w$ antipodal to the unique node of degree 4 in the graph $H_m=D_m$. Consider the first $k$ steps of this execution. Observe that $\cV(v,G_1,k)=\cV(w,H_m,k)$, by construction of graphs $G_i$ and $H_i$. By induction on the step number, in the first $k$ steps of execution $E_2$, algorithm $\cU$ calls the input enumeration algorithm for exactly the same integers as it does in execution $E_1$. For execution $E_1$  the input algorithm is 
$\cE(\cF)$ and for execution $E_2$ it is $\cE(\cF^*)$, but these integers are at most $r$ and the first $r$ graphs in the canonical enumeration of $\cF$ and of $\cF^*$ are identical. Hence the returned graphs in these calls in both executions are identical. It follows that the $k$ steps in both executions are identical. Since the agent stops after the $k$-th step of execution $E_1$, it must also stop after the $k$-th step of execution $E_2$. However, in execution $E_2$ it cannot explore the graph $H_m=D_m$ because this graph has more than $k+1$ nodes. This concludes the proof.
\end{proof}

One could hope for salvaging the idea of a universal exploration algorithm processing algorithmic input by restricting the class of input explorable families from recursively enumerable (as we did above)
to recursive. In this case the hypothetical universal algorithm could be given as input the {\em decision} algorithm that answers, for any graph $G$, if this graph belongs to the family $\cF$ of graphs that should be explored. Similarly as before, the hypothetical universal algorithm could run the decision algorithm on any graph of its choice and learn in one step if this graph belongs to the recursive family of graphs that should be explored. However, it is easy to see that even with this restriction our negative result still holds. Indeed, it is enough to modify the proof of Theorem \ref{no-univ} by defining $r$ to be the {\em largest size of a graph} for which the decision algorithm was called in the execution $E_1$. The rest of the proof remains unchanged.

\section{Conclusion}

We gave a characterization of explorable families of graphs, and provided, for any such family, a {\em dedicated} exploration algorithm that explores any graph of the family starting from any node. Then we studied the issue of the existence of a deterministic universal exploration algorithm that, for any explorable family given as input,  would explore any graph of the family starting from any node. Since such families may be infinite, it has to be made precise how would they be given to the hypothetical universal algorithm. We showed that a very liberal approach to this issue of providing the input, namely the assumption of an oracle that can answer any yes/no query asked by the universal algorithm about the input family, permits us to construct such a universal algorithm. This approach, however, is arguably unrealistic.
Hence we defined the  way of presenting the input family to the hypothetical universal algorithm in a more restrictive but more realistic way: by giving the universal algorithm either the enumeration
algorithm of the input family of graphs, in the case when this family is recursively enumerable, or giving it the decision algorithm for the input family, if the latter is recursive. We showed that this more realistic way of presenting the input explorable family precludes the existence of a deterministic universal exploration algorithm.

A similar idea to that in the proof of Theorem \ref{no-univ} can be used to prove that the problem of explorability of a family of graphs is undecidable in the following sense. There does not exist a decision algorithm that, given as input an enumeration algorithm of a recursively enumerable family $\cG$ of graphs, can decide whether the family $\cG$ is explorable. 

In this paper we concentrated on the issue of feasibility of exploration for families of graphs, rather than on the efficiency of exploration. An open, probably very challenging problem yielded by our study is to find, for any explorable family of graphs, a dedicated algorithm which would explore any graph of this family in optimal time.

\end{document}